\newcommand{\version}{arxiv}

\RequirePackage{ifthenx}
\ifthenelse{\isundefined{\version}}{\PackageError{paper.tex}{The version is not specified}{}}{}
\newcommand{\ifversion}[2]{\ifthenelse{\equal{\version}{#1}}{#2}{}}

\ifversion{lncs}{
    \documentclass[runningheads]{llncs}
}
\ifversion{arxiv}{
    \documentclass[11pt,a4paper]{article}
    \usepackage{amsthm}
    \usepackage[top=2cm, bottom=2cm, left=2cm, right=2cm, includefoot]{geometry}
    \usepackage{authblk}
    \usepackage{url}
}
\usepackage[T1]{fontenc}
\usepackage{amssymb, amsmath}
\usepackage[linesnumbered,ruled]{algorithm2e}
\usepackage{graphicx}
\usepackage{tabularx}
\usepackage{arydshln}
\usepackage{circledsteps}
\usepackage{xcolor}

\SetArgSty{textup}

\ifversion{lncs}{
    \newtheorem{prob}{Problem}
}
\ifversion{arxiv}{
    \theoremstyle{definition}
    \newtheorem{prob}{Problem}
    \theoremstyle{plain}
    \newtheorem{definition}{Definition}
    
    \newtheorem{theorem}{Theorem}
    \newtheorem{lemma}{Lemma}
    \newtheorem{corollary}{Corollary}
    \theoremstyle{remark}
    \newtheorem{remark}{Remark}
    \newtheorem{example}{Example}
}

\makeatletter
\ifversion{arxiv}{
    \renewenvironment{proof}[1][\proofname]{\par
        \normalfont\topsep6\p@\@plus6\p@\relax
        \trivlist
        \item\relax
        {\itshape#1\@addpunct{.}}\hspace\labelsep\ignorespaces
    }{%
        \popQED\endtrivlist\@endpefalse
    }
}
\makeatother

\newcommand{\mtt}[1]{\mathtt{#1}}
\newcommand{\enumbyparen}{\renewcommand{\labelenumi}{{\rm (\arabic{enumi})}}}

\newcommand{\N}{\mathbb{N}}
\newcommand{\Nplus}{\N^{+}}
\newcommand{\floor}[1]{{\lfloor #1 \rfloor}}

\newcommand{\Consts}{{\rm\Sigma}}
\newcommand{\Params}{{\rm\Pi}}
\newcommand{\Char}{{\Consts\cup\Params}}
\newcommand{\String}{{(\Char)^*}}
\newcommand{\substr}[3]{{#1[#2:#3]}}
\newcommand{\match}{\equiv}
\newcommand{\nmatch}{\not\equiv}
\newcommand{\period}[1]{{\mathit{period}(#1)}}
\newcommand{\hasperiod}[3]{{#2 \parallel_{#3} #1}}
\newcommand{\nhasperiod}[3]{{#2 \nparallel_{#3} #1}}
\newcommand{\pcs}[1]{{\Params_{#1}}}
\newcommand{\first}[1]{{\mathit{first}_{#1}}}
\newcommand{\pccount}[1]{{\mathit{count}_{#1}}}
\newcommand{\reach}[1]{{\mathit{reach}_{#1}}}

\begin{document}

\title{Efficient Parameterized Pattern Matching in Sublinear Space}

\ifversion{lncs}{
    \author{Haruki~Ideguchi \and
        Diptarama~Hendrian\orcidID{0000-0002-8168-7312} \and
        Ryo~Yoshinaka\orcidID{0000-0002-5175-465X} \and
        Ayumi~Shinohara\orcidID{0000-0002-4978-8316}}
    
    \authorrunning{H. Ideguchi et al.}

    \institute{Graduate School of Information Sciences, Tohoku University, Sendai, Japan\\
    \email{haruki.ideguchi.q3@dc.tohoku.ac.jp,\\ \{diptarama, ryoshinaka, ayumis\}@tohoku.ac.jp}}
}
\ifversion{arxiv}{
    \author{Haruki~Ideguchi}
    \author{Diptarama~Hendrian}
    \author{Ryo~Yoshinaka}
    \author{Ayumi~Shinohara}
    
    \date{}
    
    \affil{Tohoku University, Japan}
}

\maketitle

\begin{abstract}

The parameterized matching problem is a variant of string matching, which is to search for all \emph{parameterized} occurrences of a pattern $P$ in a text $T$.
In considering matching algorithms, the combinatorial natures of strings, especially \emph{periodicity}, play an important role.
In this paper, we analyze the properties of periods of parameterized strings and propose a generalization of Galil and Seiferas's exact matching algorithm (1980) into parameterized matching, which runs in $O(\pi|T|+|P|)$ time and $O(\log{|P|}+|\Params|)$ space in addition to the input space, where $\Params$ is the parameter alphabet and $\pi$ is the number of parameter characters appearing in $P$ plus one.

\ifversion{lncs}{%
    \keywords{Parameterized matching \and String matching \and Sublinear space \and Combinatorics on words}
}

\end{abstract}

\section{Introduction}

String matching is a problem to search for all occurrences of a pattern $P$ in a text $T$.
Since it is one of the most important computer applications, many efficient algorithms for the problem have been proposed.
Let us denote the length of $T$ and $P$ by $n$ and $m$, respectively.
While a naive algorithm takes $O(nm)$ time to solve the problem, Knuth, Morris, and Pratt~\cite{KMP} gave an algorithm which runs in only $O(n+m)$ time by constructing auxiliary arrays called \emph{border arrays}.
After that, various algorithms to solve the problem in linear time have been proposed, which use auxiliary data structures, such as suffix trees~\cite{SuffixTree}, suffix arrays~\cite{SuffixArray}, LCP arrays~\cite{SuffixArray}.
All of those algorithms outperform the naive algorithm in terms of time complexity.
They require additional space to store their auxiliary data, whose sizes are typically $\Theta(n)$ or $\Theta(m)$.
On the other hand, studies for reducing such extra space were conducted.
Firstly, Galil and Seiferas reduced extra space usage to $O(\log{m})$~\cite{GSlog}, and later several time-space-optimal, $O(n+m)$ time and $O(1)$ extra-space algorithms were devised~\cite{Crochemore,CP,GS}.

In this paper, we consider a variant of string matching: \emph{parameterized matching}.
It is a pattern matching paradigm in which two strings are considered a \emph{match} if we can map some characters (\emph{parameter characters}) in one string to characters in another string.
This paradigm was first introduced by Baker~\cite{p} for use in software maintenance by the ability to detect `identical' computer programs renaming their variables.
For solving the parameterized matching problem, a number of linear-time algorithms have been proposed that extend algorithms for exact matching~\cite{pKMP,pSuffixArray,pPositionHeap,pPositionHeapOnline,pSuffixTree,pDAWG,pExSuffixTrie}.
However, we know of no previous attempt to reduce extra space usage for time-efficient parameterized matching algorithms to sublinear.

The main contribution of this paper is to give a sublinear-extra-space algorithm for the parameterized matching problem by extending Galil and Seiferas's exact matching algorithm~\cite{GSlog}.
It runs in $O(|\pcs{P}|n+m)$ time and $O(\log{m}+|\Params|)$ space in addition to the input space, where $\Params$ is the set of parameter characters and $\pcs{P}$ is the non-empty\footnote{We can assume $\pcs{P}\neq\emptyset$ without loss of generality. See Remark~\ref{remark:pcsnotempty}.} set of parameter characters appearing in $P$.

In order to provide the basis for our algorithm, we also investigate the properties of periodicity of parameterized strings in this paper.
It is widely known that periods of strings are useful for exact matching algorithms~\cite{Crochemore,CP,GSlog,GS,KMP}, which is also the case for parameterized matching~\cite{pKMP}.
We extend previous work on parameterized periods by Apostolico and Giancarlo~\cite{pPeriod} and derive several properties for our algorithm.
In particular, we focus on `sufficiently short' periods of parameterized strings having properties useful for matching algorithms.
Those results contain a parameterized version of Fine and Wilf's periodicity lemma~\cite{period}.

\begin{remark}
    The time and space complexities of our algorithm stated above are based on a computing model in which functions $\Params\to\N$ can be stored as arrays.
    If not, one can use AVL trees~\cite{AVL} instead of arrays to store such functions.
    Then, our algorithm runs in $O((|\pcs{P}|n+m)\log|\pcs{P}|)$ time and $O(\log{m}+|\pcs{P}|)$ extra space.
\end{remark}

\section{Preliminaries}

Let $\N$ and $\Nplus$ be the set of natural numbers including and excluding 0, respectively.
For $x,y\in\N$, we denote by $x\mid y$ that $y$ is a multiple of $x$.

\subsection{Parameterized Matching Problem}

In parameterized matching, we consider two disjoint alphabets: the \emph{constant alphabet} $\Consts$ and the \emph{parameter alphabet} $\Params$.
A string over $\Char$ is called a \emph{parameterized string} or a \emph{p-string}.
Consider a p-string $w\in\String$.
We denote the length of $w$ by $|w|$.
For $0 \le i < |w|$, let us denote $i$-th letter of $w$ by $w[i]$, where the index $i$ is 0-based.
For $0 \le i \le j \le |w|$, we denote the substring $w[i]w[i+1]\cdots w[j-1]$ by $\substr{w}{i}{j}$.
(Note that $\substr{w}{i}{j}$ does not contain $w[j]$.)

We denote the set of permutations of $\Params$ by $S_\Params$.
Throughout this paper, for a permutation $f\in S_\Params$ and a constant character $c\in\Consts$, let $f(c)=c$.
Then, the map $f$ is naturally expanded as a bijection over p-strings: $\String\to\String$.

\begin{definition}[Baker \cite{p}]
    \label{def:match}
    Two p-strings $x$ and $y$ are called a \emph{parameterized-match} or a \emph{p-match} if and only if there exists a permutation $f\in S_\Params$ such that $f(x)=y$.
    Denote this relation by $x\match y$.
\end{definition}

\begin{example}
    Let $\Consts = \{ \mathtt{a}, \mathtt{b}, \mathtt{c} \}$ and $\Params = \{ \mathtt{A}, \mathtt{B}, \mathtt{C} \}$.
    We have $\mathtt{ABaCBCa} \match \mathtt{BCaACAa}$ with a witness $f$ such that $f(\mathtt{A})=\mathtt{B}$, $f(\mathtt{B})=\mathtt{C}$, and $f(\mathtt{C})=\mathtt{A}$.
\end{example}

Clearly, the relation $\match$ is an equivalence relation over $\String$.
Note that if $x\match y$, we have $|x|=|y|$ and $\substr{x}{i}{j}\match \substr{y}{i}{j}$ for any $0\le i\le j\le|x|$.
By this relation, the problem we consider in this paper, the \emph{parameterized matching problem}, is defined as follows.
\begin{prob}[\cite{p}]
    \label{problem}
    Given two p-strings $T$ (text) and $P$ (pattern), find all $0\le i\le|T|-|P|$ such that $\substr{T}{i}{i+|P|}\match P$.
\end{prob}

\subsection{Periodicity of Parameterized Strings}

Periodicity is one of the most fundamental concepts in combinatorics of strings and a wealth of applications.
In exact matching, the Knuth-Morris-Pratt algorithm and various algorithms based on it rely on the properties of periods \cite{KMP,GSlog,GS,CP,Crochemore}.
It is also the case for parameterized matching \cite{pKMP}, where periods of parameterized strings are defined as follows:

\begin{definition}[Apostolico and Giancarlo \cite{pPeriod}]
    \label{def:p}
    Consider $w\in\String$ and $p\in\Nplus$ with $p\le|w|$.
    Then, $p$ is called a \emph{period} of $w$ if and only if $\substr{w}{0}{|w|-p}\match\substr{w}{p}{|w|}$.
\end{definition}

If $p$ is a period of $w$, there exists $f\in S_\Params$ satisfying $f(\substr{w}{0}{|w|-p})=\substr{w}{p}{|w|}$ by definition.
We denote this relation by $\hasperiod{w}{p}{f}$ or simply by $\hasperiod{w}{p}{}$ when $f$ is not specified.

In general, a p-string $w$ can have multiple periods.
We denote the shortest period of $w$ as $\period{w}$.
It is clear that a period $p$ of a p-string $w$ is also a period of any substring $w'$ of $w$ such that $|w'|\ge p$.

\begin{example}
    Let $\Consts = \{ \mathtt{a}, \mathtt{b}, \mathtt{c} \}$ and $\Params = \{ \mathtt{A}, \mathtt{B}, \mathtt{C} \}$.
    For $w := \mathtt{ABaCBCaACAa}$, we have $\hasperiod{w}{4}{f}$ as $\mathtt{ABaCBCa} \match \mathtt{BCaACAa}$ with $f(\mathtt{A})=\mathtt{B}$, $f(\mathtt{B})=\mathtt{C}$, and $f(\mathtt{C})=\mathtt{A}$.
\end{example}

Instead of Definition~\ref{def:pp}, one can use the following equivalent definition for periods, which is a more intuitive representation of the repetitive structure of strings:

\begin{lemma}[\cite{pPeriod}]
    \label{lem:pdisp}
    Consider $w\in\String$, $p\in\Nplus$, and $f\in S_\Params$.
    Then, $\hasperiod{w}{p}{f}$ holds if and only if $w$ can be written as
    \begin{equation*}
        w=f^0(v) \cdot f^1(v) \cdot f^2(v) \cdots f^\floor{\rho}(v) \cdot f^{\floor{\rho}+1}(v'),
    \end{equation*}
    where $\rho=\frac{|w|}{p}$, $v=\substr{w}{0}{p}$ and $v'$ is a prefix of $v$ (allowing the case $v'$ is empty).
\end{lemma}

The following lemma has important applications for various matching algorithms.
Particularly, it is used to shift the pattern string safely in the Knuth-Morris-Pratt algorithm and variants~\cite{pKMP,KMP}.

\begin{lemma}
    \label{lem:shift}
    Consider $x,y\in\String$ with $x \match y$.
    For any $0 < \delta < \period{y}$, we have $\substr{x}{\delta}{|x|} \nmatch \substr{y}{0}{|y|-\delta}$.
\end{lemma}
\begin{proof}
    We give a proof by contraposition.
    Suppose $\substr{x}{\delta}{|x|} \match \substr{y}{0}{|y|-\delta}$.
    Then we have \ifversion{arxiv}{\linebreak}$\substr{y}{0}{|y|-\delta} \match \substr{x}{\delta}{|x|} \match \substr{y}{\delta}{|y|}$, which means $\hasperiod{y}{\delta}{}$.
    Hence, $\delta \ge \period{y}$ holds.
    \qed
\end{proof}

One of the main interest regarding string periodicity is what holds when a string $w$ has two different periods $p$ and $q$.
For ordinary strings, Fine and Wilf's periodicity lemma \cite{period} gives an answer: $\gcd(p,q)$ is also a period when $|w|\ge p+q-\gcd(p,q)$, where $\gcd(p,q)$ is the greatest common divisor of $p$ and $q$.
Apostolico and Giancarlo showed a similar property for parameterized strings.

\begin{lemma}[\cite{pPeriod}]
    \label{lem:p_gcd}
    For $w\in\String$, $p,q\in\Nplus$, and $f,g\in S_\Params$, assume that $\hasperiod{w}{p}{f}$ and $\hasperiod{w}{q}{g}$.
    If $|w|\ge p+q$ and $fg=gf$, we have $\hasperiod{w}{\gcd(p,q)}{}$.
\end{lemma}

It is known that the length $|w|=p+q-\gcd(p,q)$ is not sufficient for this lemma unlike ordinary strings \cite{pPeriod}.

\section{Properties of Parameterized Periods}

In this section, we show some properties of periods of parameterized strings.
They play an important role in our algorithm presented in Section~\ref{section:algorithm}.

\subsection{Alternative Periodicity Lemma}

The requirements of Lemma~\ref{lem:p_gcd} are slightly different from Fine and Wilf's lemma for ordinary strings.
Particularly, the commutativity of $f$ and $g$ is essential~\cite[Lemma~5]{pPeriod}.
We show in this section a new periodicity lemma for parameterized strings which does not assume the commutativity.

Firstly, we focus on parameter characters contained in a given p-string and its substrings.
For $w\in\String$, we denote by $\pcs{w}$ the set of parameter characters appearing on $w$.

\begin{example}
    Let $\Consts = \{ \mathtt{a}, \mathtt{b}, \mathtt{c} \}$ and $\Params = \{ \mathtt{A}, \mathtt{B}, \mathtt{C} \}$.
    For $w:=\mathtt{ABabAca}$, we have $\pcs{w}=\{ \mathtt{A}, \mathtt{B} \}$.
\end{example}



\begin{lemma}
    \label{lem:chars}
    Consider $w\in\String$ and any of its substrings $w'$ and $w''$.
    Then, the followings hold:
    \begin{itemize}
        \item If $|w'|\ge \period{w}\cdot(|\pcs{w}|-1)$, we have $|\pcs{w'}|\ge|\pcs{w}|-1$.
        \item If $|w''|\ge \period{w}\cdot|\pcs{w}|$, we have $\pcs{w''}=\pcs{w}$.
    \end{itemize}
\end{lemma}

\begin{proof}
    The case $\pcs{w} = \emptyset$ is trivial.  Suppose $\pcs{w} \neq \emptyset$.
    Let $p := \period{w}$ and $f$ be a permutation of $\Params$ such that $\hasperiod{w}{p}{f}$.
    It suffices to show the lemma for the cases $|w'|=p\cdot(|\pcs{w}|-1)$ and $|w''|=p\cdot|\pcs{w}|$.
    By Lemma~\ref{lem:pdisp}, $w'$ and $w''$ can be written as $w'=v'\cdot f(v') \cdots f^{|\pcs{w}|-2}(v')$ and $w''=v''\cdot f(v'') \cdots f^{|\pcs{w}|-1}(v'')$, where $v'$ and $v''$ are the prefixes of $w'$ and $w''$ of length $p$, respectively.
    Now, we consider the cyclic decomposition of $f$.

    Suppose the characters in $\pcs{w}$ make one cyclic permutation in $f$.
    Let $a$ be any parameter character contained in $v'$.
    Note that $a,f(a),\cdots,f^{|\pcs{w}|-2}(a)$ are all different characters and appear in $w'$.
    Therefore, we have $|\pcs{w'}|\ge|\pcs{w}|-1$.
    The analogous argument shows $|\pcs{w''}| = |\pcs{w}|$.

    Suppose the characters in $\pcs{w}$ make two or more cyclic permutations in $f$.
    Then, those cyclic permutations are all of length $|\pcs{w}|-1$ or less.
    For $0\le i<|w|$, there exists an integer $k$ such that $w[i+kp], w[i+(k+1)p], \cdots, w[i+(k+|\pcs{w}|-2)p]$ are all contained in $w'$.
    Then, those characters can be represented as $f^{k}(w[i]),f^{k+1}(w[i]),\cdots,f^{k+|\pcs{w}|-2}(w[i])$, and by the assumption about $f$, at least one of them is equal to $w[i]$.
    Therefore, we have $w[i]\in\pcs{w'}$.
    Since $i$ is arbitrary, we end up with $\pcs{w}\subseteq\pcs{w'}$, as required.
    \qed
\end{proof}

Now, we show a variant of Lemma~\ref{lem:p_gcd}.
It does not require any assumption on the permutations, in exchange of a stricter requirement for the length of strings.

\begin{lemma}
    \label{lem:p_gcd2}
    Suppose $w\in\String$ with $\pcs{w}\neq\emptyset$ has periods $p$ and $q$.
    If $|w|\ge p+q+\min(p,q)\cdot(|\pcs{w}|-1)$, we have $\hasperiod{w}{\gcd(p,q)}{}$.
\end{lemma}

\begin{proof}
    Let $f$ and $g$ be permutations of $\Params$ such that $\hasperiod{w}{p}{f}$ and $\hasperiod{w}{q}{g}$.
    Without loss of generality, we suppose $f(a)=a$ and $g(a)=a$ for any $a\in\Params\setminus\pcs{w}$.
    By Lemma~\ref{lem:p_gcd}, it suffices to show that $fg=gf$.
    Let $w':=\substr{w}{0}{|w|-p-q}$.
    Then, notice that $fg(w') = f(\substr{w}{q}{|w|-p}) = \substr{w}{p+q}{|w|} = g(\substr{w}{p}{|w|-q}) = gf(w')$, which claims $fg(a)=gf(a)$ for any $a\in\pcs{w'}$.
    Moreover, given $|w'|=|w|-p-q \ge \min(p,q)\cdot(|\pcs{w}|-1) \ge \period{w}\cdot(|\pcs{w}-1|)$, we have $|\pcs{w'}|\ge|\pcs{w}|-1$ by Lemma~\ref{lem:chars}.
    Hence, the permutations $fg$ and $gf$ behave the same for at least $|\Params|-1$ parameter characters.
    This implies $fg=gf$.
    \qed
\end{proof}

\begin{corollary}
    \label{cor:p}
    Suppose $w\in\String$ with $\pcs{w}\neq\emptyset$ has a period $q$.
    If $q\le\frac{|w|}{|\pcs{w}|+1}$, then $\period{w} \mid q$.
\end{corollary}

\begin{proof}
    Let $p:=\period{w}$.
    By $p\le q\le \frac{|w|}{|\pcs{w}|+1}$, we have $p\cdot|\pcs{w}|+q \le q\cdot(|\pcs{w}|+1) \ifversion{lncs}{\linebreak}\le \frac{|w|}{|\pcs{w}|+1}(|\pcs{w}|+1)=|w|$.
    Hence, we can use Lemma~\ref{lem:p_gcd2} to obtain $\hasperiod{w}{\gcd(p,q)}{}$.
    Then, since $p$ is the smallest period of $w$, we have $\gcd(p,q)\ge p$, which means $\gcd(p,q)=p$ i.e. $p \mid q$, as required.
    \qed
\end{proof}

\subsection{Prefix Periods}

Galil and Seiferas's exact matching algorithm~\cite{GSlog} can be regarded as an extension of the Knuth-Morris-Pratt algorithm~\cite{KMP}.
The main idea of their algorithm is to deal with only periods of pattern prefixes which are `short enough.'
They pointed out that periods shorter than $\frac{1}{k}$ times of the length of the string have useful properties for saving space usage in exact string matching for an arbitrarily fixed $k \ge 3$.
We show in this section that similar properties hold for parameterized strings as well when $k$ is set to be $|\pcs{w}|+2$.
Most part of those properties come from Lemma~\ref{lem:p_gcd2} we proved in the previous section.

\begin{lemma}
    \label{prop:unique}
    Suppose $w\in\String$ has a period $p$.
    If $p\le\frac{|w|}{|\pcs{w}|+1}$, there exists only one character $a\in\Char$ such that $\hasperiod{wa}{p}{}$.
\end{lemma}

\begin{proof}
    Consider the prefix $w':=\substr{w}{0}{|w|-p}$.
    By $p\le\frac{|w|}{|\pcs{w}|+1}$, we have $|w|-p \ge p|\pcs{w}| \ge \period{w}|\pcs{w}|$.
    By Lemma~\ref{lem:chars}, $\pcs{w'} = \pcs{w}$.
    Therefore, $w[|w|-p]$ already appears in $w'$ as $w[i]=w[|w|-p]$ for some $i < |w|-p$.
    Hence, for any $f$ such that $\hasperiod{w}{p}{f}$, it holds that $\hasperiod{wa}{p}{f}$ if and only if $a = w[i+p]$.
    \qed
\end{proof}

\begin{corollary}
    \label{cor:union}
    Suppose $w\in\String$ has a period $p$.
    For any $\ell\in\Nplus$ with $\ell p\le\frac{|w|}{|\pcs{w}|+1}$, we have $\hasperiod{wa}{p}{} \iff \hasperiod{wa}{\ell p}{}$ for any $a\in\Char$.
\end{corollary}

\begin{proof}
    By Lemma~\ref{prop:unique}, the characters $a_1$ and $a_2$ such that $\hasperiod{wa_1}{p}{}$ and $\hasperiod{wa_2}{\ell p}{}$ are unique respectively.
    Then, since $\hasperiod{wa_1}{p}{} \implies \hasperiod{wa_1}{\ell p}{}$ (shown immediately by Lemma~\ref{lem:pdisp}), we get $a_1=a_2$, as required.
    \qed
\end{proof}

Now, we introduce the key concept for our algorithm: \emph{prefix periods}.
This is a natural extension of the one introduced in~\cite{GS} for parameterized strings.
Hereafter in this section, we consider a fixed p-string $w\in\String$ with $\pcs{w}\neq\emptyset$ and let $k:=|\pcs{w}|+2$.

\begin{definition}
    \label{def:pp}
    A positive integer $p\in\Nplus$ is called a \emph{prefix period} of $w$ if and only if there exists a prefix $w'$ of $w$ such that $period(w')=p$ and $p\le\frac{|w'|}{k}$.
\end{definition}

\begin{table}[t]
    \centering
    \caption{
        Let $\Params=\{\mathtt{A},\mathtt{B}\}$.
        A p-string $w:=\texttt{ABABBABAABABBABAABBA}$ has prefix periods 1 and 4.
        Circled numbers in the table below are prefix periods of $w$ with $\substr{w}{0}{i+1}$ as witnesses.
        For instance, 4 is a prefix period of $w$ with $\substr{w}{0}{18}$ as a witness because $\period{\substr{w}{0}{18}}=4$ and $4 \le \frac{|\substr{w}{0}{18}|}{k}$.
        (Note that $k = |\pcs{w}|+2 = 4$.)
    }
    \label{table:pp_example}
    \newcommand{\cellwidth}{0.03\textwidth}
    \ifversion{arxiv}{\setlength{\tabcolsep}{0.005\textwidth}}
    \begin{tabular}{
        |c|
        *{4}{>{\centering\arraybackslash}p\cellwidth}
        ;{0.4mm/0.7mm}
        *{14}{>{\centering\arraybackslash}p\cellwidth}
        ;{0.4mm/0.7mm}
        *{2}{>{\centering\arraybackslash}p\cellwidth}|
    }
        \hline
        $i$
            & 0 & 1 & 2 & 3
            & 4 & 5 & 6 & 7
            & 8 & 9 & 10 & 11
            & 12 & 13 & 14 & 15
            & 16 & 17 & 18 & 19
        \\
        $w[i]$
            & \texttt{A} & \texttt{B} & \texttt{A} & \texttt{B}
            & \texttt{B} & \texttt{A} & \texttt{B} & \texttt{A}
            & \texttt{A} & \texttt{B} & \texttt{A} & \texttt{B}
            & \texttt{B} & \texttt{A} & \texttt{B} & \texttt{A}
            & \texttt{A} & \texttt{B} & \texttt{B} & \texttt{A}
        \\
        $\period{\substr{w}{0}{i+1}}$
            & 1 & 1 & 1 & \Circled{1}
            & 4 & 4 & 4 & 4
            & 4 & 4 & 4 & 4
            & 4 & 4 & 4 & \Circled{4}
            & \Circled{4} & \Circled{4} & 18 & 18
        \\
        \hline
            \multicolumn{5}{r|}{$\reach{w}(1)=4 \,$} &
            \multicolumn{14}{r|}{$\reach{w}(4)=18 \,$} &
            \multicolumn{2}{r}{}
        \\
    \end{tabular}
\end{table}

We give an example of prefix periods in Table~\ref{table:pp_example}.
For a fixed $p$, only prefixes $w'$ of $w$ satisfying $|w'|\ge kp$ can be a witness for $p$ being a prefix period.
We show in the following lemmas that it suffices to consider only one prefix $w'=\substr{w}{0}{kp}$ for checking whether $p$ is a prefix period.

\begin{lemma}
    \label{lem:prest}
    For any $a\in\Char$, if $\period{wa}\neq\period{w}$, we have $\period{wa} > \frac{|w|}{|\pcs{w}|+1}$.
\end{lemma}

\begin{proof}
    We show the lemma by contraposition.
    Suppose $\period{wa} \le \frac{|w|}{|\pcs{w}|+1}$.
    Since $\period{wa}$ is also a period of $w$, we can use Corollary~\ref{cor:p} to obtain \ifversion{lncs}{\linebreak}$\period{w} \mid \period{wa}$.
    Therefore, we get $\hasperiod{wa}{\period{w}}{}$ by Corollary~\ref{cor:union}, which implies $\period{w}\ge\period{wa}$.
    On the other hand, we have $\period{w}\le\period{wa}$ by definition.
    Thus $\period{w}=\period{wa}$ holds.
    \qed
\end{proof}

\begin{lemma}
    \label{lem:ppdef}
    Consider any $0<p\le\frac{|w|}{k}$.
    Then, $p$ is a prefix period of $w$ if and only if $\period{w'}=p$ where $w':=\substr{w}{0}{kp}$.
\end{lemma}

\begin{proof}
    $(\!\impliedby\!)\;$
    Immediate by the definition of prefix periods.
    \\
    $(\!\implies\!)\;$
    Let $v$ be a prefix of $w$ that witnesses $p$ being a prefix period, i.e., $|v| \ge kp$ and $\period{v} = p$.
    If $|v| = kp$, we are done.
    Suppose $|v|>kp$ and let $u:=\substr{v}{0}{|v|-1}$.
    Then,
    \(
        \period{v} = p < \frac{|v|}{k} \le \frac{|v|}{|\Pi_u|+2} < \frac{|u|}{|\Pi_u|+1} 
    \).
    By Lemma~\ref{lem:prest}, we have $\period{u}=\period{v}=p$.
    By repeatedly applying this discussion, we can shorten the witness up to length $kp$.
    \qed
\end{proof}

Next, we introduce an auxiliary function $\reach{w}$.

\begin{definition}
    \label{def:reach}
    For any $0<p\le|w|$, let
    \begin{equation*}
        \reach{w}(p) := \max\{r\in\N : r\le|w| \text{ \rm and } \hasperiod{\substr{w}{0}{r}}{p}{}\}.
    \end{equation*}
\end{definition}

Note that $\hasperiod{\substr{w}{0}{r}}{p}{} \iff \reach{w}(p)\ge r$ holds by definition.
Using $\reach{w}$, we get an equivalent definition of prefix periods as follows, which is directly used in our searching algorithm.

\begin{lemma}
    \label{lem:ppreach}
    Consider any $0<p\le\frac{|w|}{k}$.
    Then, $p$ is a prefix period of $w$ if and only if all the followings hold:
    \begin{enumerate}
        \enumbyparen
        \item $\reach{w}(p) \ge kp$,
        \item $\reach{w}(q) < \reach{w}(p)$ for any $0<q<p$.
    \end{enumerate}
\end{lemma}

\begin{proof}
    $(\!\implies\!)\;$
    (1) is by definition. We show (2).
    By Lemma~\ref{lem:ppdef}, $\period{\substr{w}{0}{kp}}=p$.
    Thus, $q < p$ is not a period of $\substr{w}{0}{kp}$, i.e., $\reach{w}(q) < kp \le \reach{w}(p)$ by (1).
    \vspace{0.5mm} \\
    $(\!\impliedby\!)\;$
    Let $w':=\substr{w}{0}{\reach{w}(p)}$.
    (2) implies $\period{w'}=p$ since any $q$ satisfying $0<q<p$ is not a period of $w'$.
    Additionally, we have $p\le\frac{|w'|}{k}$ by (1).
    Thus $p$ is a prefix period of $w$ with $w'$ as a witness.
    \qed
\end{proof}

Galil and Seiferas~\cite[Corollary~1]{GSlog} pointed out that the number of prefix periods of a word $w$ is $O(\log|w|)$.
We show in the following lemma that it is the case for parameterized strings.
It contributes directly to reducing the space complexity of our algorithm.

\begin{lemma}
    \label{lem:ppdouble}
    Suppose $w$ has prefix periods $p$ and $q$.
    If $p<q$, then $2p\le q$.
\end{lemma}

\begin{proof}
    We prove the theorem by contradiction.
    Suppose $p<q<2p$.
    By definition, $\hasperiod{\substr{w}{0}{kp}}{p}{}$ and $\hasperiod{\substr{w}{0}{kq}}{q}{}$ hold.
    Let $w' := \substr{w}{0}{kp}$.
    By Lemma~\ref{lem:ppdef}, $p$ is the shortest period of $w'$.
    Since both $p$ and $q$ are periods of $w'$ and $p\cdot|\pcs{w'}|+q < p\cdot|\pcs{w}|+2p = kp = |w'|$, we get $\hasperiod{w'}{\gcd(p,q)}{}$ by Lemma~\ref{lem:p_gcd2}.
    Hence, we have $\gcd(p,q)\ge  \period{w'} =p$, which claims $\gcd(p,q)=p$ i.e. $p \mid q$.
    However, this contradicts to the assumption $p < q < 2p$.
    \qed
\end{proof}

\begin{corollary}
    \label{cor:ppnum}
    The number of prefix periods of $w\in\String$ is at most $\log_2{|w|}$.
\end{corollary}

\section{Proposed Algorithm}
\newcommand{\varpp}{\mathit{PP}}
\newcommand{\varppval}[1]{\varpp[#1].\mathit{val}}
\newcommand{\varppreach}[1]{\varpp[#1].\mathit{reach}}
\newcommand{\varppidx}{\mathit{idx}}
\newcommand{\varmaxreach}{\mathit{max\hspace{-0.1em}\_\hspace{-0.1em}reach}}

\label{section:algorithm}

In this section, we propose a sublinear-extra-space algorithm for the parameterized matching problem.
Throughout this section, let $T$ and $P$ be p-strings whose lengths are $n$ and $m$ respectively, and let $k:=|\pcs{P}|+2$.
Besides, we suppose $\pcs{P}\neq\emptyset$.
Our algorithm is an extension of Galil and Seiferas's exact string matching algorithm~\cite{GSlog} and runs in $O(|\pcs{P}|n+m)$ time and $O(\log{m}+|\Params|)$ extra space.
When $|\Params|=|\pcs{P}|=1$, our algorithm behaves exactly as theirs.

\begin{remark}
    \label{remark:pcsnotempty}
    We can assume $\pcs{P}\neq\emptyset$ without loss of generality.
    When $\pcs{P}=\emptyset$, choose any $c\in\Consts$ appearing in $P$ and let constant and parameter alphabets be $\Char\setminus\{c\}$ and $\{c\}$, respectively.
\end{remark}

Firstly, we introduce a method for testing whether two p-strings match.
While it is common to use the \emph{prev-encoding}~\cite{p} for this purpose, it is not suitable for our goal since it requires additional space proportional to the input size.
Thus we use an alternative method as follows, which requires only $O(|\Params|)$ extra space.

\begin{lemma}
    \label{lem:match}
    Consider a prefix $x$ of $P$ and $y\in\String$ with $x\match y$ and any $a,b\in\Char$.
    We have $xa \match yb$ if and only if one of the followings holds:
    \begin{enumerate}
        \item $a\in\Consts$ and $a=b$,
        \item $a\in\Params$ and $\first{P}(a)\ge|x|$ and $b\in\Params$ and $\pccount{y}(b)=0$,
        \item $a\in\Params$ and $\first{P}(a)<|x|$ and $y[\first{P}(a)]=b$,
    \end{enumerate}
    where $\first{P} : \Params\to\N$ and $\pccount{y} : \Params\to\N$ are defined as follows:
    \begin{align*}
        \first{P}(c) &= \begin{cases}
            \min\{i\in\N : i<|x| \text{ and } x[i]=c\} & \text{ if $c\in\pcs{P}$,} \\
            |P| & \text{ if $c\in\Params\setminus\pcs{P}$,}
        \end{cases} \\
        \pccount{y}(c) &= |\{i\in\N : i<|y| \text{ and } y[i]=c\}|
    \end{align*}
\end{lemma}
\begin{proof}
    By definition, we have $xa \match yb$ if and only if $b=f(a)$, where $f$ satisfies $x=f(y)$.
    If $a$ is a constant character or appears in $x$, the value $f(a)$ is determined (Case~1 and 3).
    Otherwise, $b$ must be a parameter character not appearing in $y$ (Case~2).
\qed\end{proof}

Let $\mathtt{MATCH}(x,y,a,b,\first{P},\pccount{y})$ be the function which returns whether $xa \match yb$ under the condition $x \match y$ using Lemma~\ref{lem:match}.
Clearly, one can compute it in constant time if $\first{P}$ and $\pccount{y}$ are given as arrays.
Note that $\first{P}$ can be computed in $O(m)$ time and $O(|\Params|)$ space.

\subsection{Pattern Preprocessing}

In this section, we show the preprocessing for the pattern $P$ for our matching algorithm.
The output of the preprocessing is the list of pairs of a prefix period of $P$ (in ascending order) and its reach, just like Galil and Seiferas~\cite{GSlog} introduced for exact string matching.
The list plays a similar role to the \emph{border array} in the parameterized Knuth-Morris-Pratt algorithm~\cite{pKMP}.
While border array uses $\Theta(m)$ space to memorize the shortest periods of all prefixes of $P$, the prefix period list requires only $O(\log{m})$ space by Corollary~\ref{cor:ppnum}.

\begin{algorithm}[!t]
    \DontPrintSemicolon

    \caption{\texttt{PREFIX\_PERIODS}}
    \label{algo:pp}

    \KwIn{$P\in\String$ and $\mathit{first}$}
    \KwOut{a list of all prefix periods of $P$ and their reaches}
    
    \Begin{
        $k \gets |\pcs{P}|+2$ \;
        $\varpp \gets \text{empty list}$ \;
        $\varppidx \gets -1$ \;
        $(p,r) \gets (1,1)$ \;
        Set $\mathit{count}[a] \gets 0$ for each $a\in\Params$ \;
        $\varmaxreach \gets 0$ \;
        \While{$kp\le|P|$}{
            \label{algo:pp:outerwhile}
            \While{$\mathtt{MATCH}(\substr{P}{0}{r-p}, \substr{P}{p}{r}, P[r-p], P[r], \mathit{first}, \mathit{count})$}{
                \label{algo:pp:while}
                Increment $\mathit{count}[P[r]]$ \;
                \label{algo:pp:while:content}
                \label{algo:pp:count}
                $r \gets r+1$ \;
                \label{algo:pp:reach}
                \lIf{$\varppidx+1 < |\varpp|$ and $\varppval{\varppidx + 1} \le \frac{r-p}{k}$}{Increment $\varppidx$}
                \label{algo:pp:pointer:inc}
                \label{algo:pp:while:content_end}
            }
            \label{algo:pp:while_end}

            \vspace{1em}

            \If{$r \ge kp$ and $r > \varmaxreach$\label{algo:pp:detect}}{Push $(p,r)$ into $\varpp$}
            $\varmaxreach \gets \max\{\varmaxreach, r\}$ \;

            \vspace{1em}

            \eIf{$0 \le \varppidx < |\varpp|$ and $\varppreach{\varppidx} \ge r-p > 0$}{
                \label{algo:pp:branch}
                \lFor{$p \le i< p+\varppval{\varppidx}$}{Decrement $\mathit{count}[P[i]]$}
                \label{algo:pp:count:dec}
                $p \gets p+\varppval{\varppidx}$ \;
                \label{algo:pp:shift}
            }{
                \lFor{$p \le i< r$}{Decrement $\mathit{count}[P[i]]$}
                \label{algo:pp:count:reset}
                $p \gets p+\floor{\frac{r-p}{k}}+1$ \;
                \label{algo:pp:limit_shift}
                $r \gets p$ \;
                \label{algo:pp:limit_shift_end}
            }
            \KwSty{until} $\varppval{\varppidx} \le \frac{r-p}{k}$ or $\varppidx=-1$ \KwSty{do} Decrement $\varppidx$ \;
            \label{algo:pp:pointer:dec}

            \label{algo:pp:branch_end}
        }

        \vspace{0.6em}

        \Return{$\varpp$} \;
    }
\end{algorithm}

We present the algorithm for computing the list of pairs of a prefix period and its reach in Algorithm~\ref{algo:pp}.
The algorithm finds prefix periods and their reaches in order from the smallest to the largest and put them into the list $\varpp$.
By $\varppval{\varppidx}$ and $\varppreach{\varppidx}$, we denote the $\varppidx$-th prefix period and its reach in $\varpp$, respectively.
Starting with $p=1$, it monotonically increases $p$ and checks whether an integer $p$ is a prefix period based on Lemma~\ref{lem:ppreach}.
Throughout the algorithm run, we maintain the invariant 
\begin{equation*}
    \hasperiod{\substr{P}{0}{r}}{p}{}\text{, i.e., }\substr{P}{0}{r-p} \equiv \substr{P}{p}{r}  \tag{$\spadesuit$}  
\end{equation*}
We calculate $\reach{P}(p)$ by increasing $r$ as long as $\substr{P}{0}{r-p} \equiv \substr{P}{p}{r}$ holds (Lines~\ref{algo:pp:while}--\ref{algo:pp:while_end}).
To let the function $\mtt{MATCH}$ decide $\substr{P}{0}{r-p} \equiv \substr{P}{p}{r}$, we use two auxiliary arrays $\first{}$ and $\pccount{}$ that satisfy $\first{}[a]=\first{P}(a)$ and $\pccount{}[a] = \pccount{P[p:r]}(a)$, defined in Lemma~\ref{lem:match}.
Moreover, we maintain the variable \ifversion{lncs}{\linebreak}$\varmaxreach$ to be the largest reach calculated so far.
By Lemma~\ref{lem:ppreach}, the condition of Line~\ref{algo:pp:detect} is satisfied if and only if $p$ is a prefix period.

One can construct the list $\varpp$ by incrementing $p$ one by one, but it takes too much time.
Instead, we skip calculating $\reach{P}(p)$ if we are sure that $p$ is not a prefix period.
For realizing an efficient shift, we maintain a variable $\varppidx$ so that it points at the largest index of $\varpp$ such that $\varppval{\varppidx} \le \frac{r-p}{k}$ (Lemma~\ref{lem:algo:pp:index} below).
The shift amount is determined in the following manner.
If $\varppreach{\varppidx} \ge r-p > 0$ at Line~\ref{algo:pp:branch}, Lemmas~\ref{lem:algo:pp:rp} and~\ref{lem:algo:pp:p} imply $\varppval{\varppidx} = \period{\substr{P}{0}{r-p}}$.
Hence, Lemma~\ref{lem:ppskip} justifies the shift amount $\varppval{\varppidx}$ of $p$ at Line~\ref{algo:pp:shift}.
On the other hand, if $\varppreach{\varppidx} < r-p$, by Lemma~\ref{lem:algo:pp:rp}, we have $\period{\substr{P}{0}{r-p}} > \frac{r-p}{k}$.
This justifies the shift $\floor{\frac{r-p}{k}}+1$ of $p$ at Line~\ref{algo:pp:limit_shift} again by Lemma~\ref{lem:ppskip}.
If $r-p=0$, then $p$ is incremented by just one.

We now prove the lemmas used for justifying our algorithm behavior in the above discussion.
Firstly, we assume the invariant $\spadesuit$.
\begin{lemma}
    \label{lem:algo:pp:index}
    Throughout Algorithm~\ref{algo:pp}, the value of the variable $\varppidx$ is always the upper bound that satisfies $\varppval{\varppidx} \le \frac{r-p}{k}$.
    If there does not exists such index, we have $\varppidx=-1$.
\end{lemma}
\begin{proof}
    The variable $\varppidx$ is updated in conjunction with $p$ and $r$ to preserve the condition.
    See Lines~\ref{algo:pp:pointer:inc} and \ref{algo:pp:pointer:dec}.
    \qed
\end{proof}

\begin{lemma}
    \label{lem:algo:pp:p}
    Let $\spadesuit$ hold at Line~\ref{algo:pp:branch} in Algorithm~\ref{algo:pp}.
    If $\period{\substr{P}{0}{r-p}} \le \frac{r-p}{k}$, we have $\varppval{\varppidx} = \period{\substr{P}{0}{r-p}}$.
\end{lemma}
\begin{proof}
    Let $w' := \substr{P}{0}{r-p}$, $p' := \period{w'}$, $p'' := \varppval{\varppidx}$, and $w'' = \substr{P}{0}{kp''}$.
    By the assumption, $p'$ is a prefix period of $P$.
    Additionally, we have $p' \le p$ since $\hasperiod{w'}{p}{}$.
    Thus $p'$ is in the list $\varpp$, and thus we have $p' \le p''$ by Lemma~\ref{lem:algo:pp:index}.
    On the other hand, we have $\period{w''}=p''$ by Lemma~\ref{lem:ppdef}.
    Since $|w''|=kp'' \le r-p=|w'|$, we have $\period{w''} \le \period{w'}$, i.e. $p'' \le p'$.
    Hence we get $p' = p''$.
    \qed
\end{proof}

\begin{lemma}\label{lem:algo:pp:rp}
    Let $\spadesuit$ hold at Line~\ref{algo:pp:branch} in Algorithm~\ref{algo:pp}.
    We have $\varppreach{\varppidx} \ge r-p \iff\ifversion{arxiv}{\makebox[0.04\textwidth]{}\linebreak} \period{\substr{P}{0}{r-p}} \le \frac{r-p}{k}$.
\end{lemma}
\begin{proof}
    Let $w' := \substr{P}{0}{r-p}$ and $p' := \varppval{\varppidx}$.
    \vspace{0.5mm}
    \\
    $(\!\implies\!)\;$
    We have $\hasperiod{w'}{p'}{}$ by the assumption.
    Then $\period{w'} \le p' \le \frac{r-p}{k}$ holds by Lemma~\ref{lem:algo:pp:index}.
    \\
    $(\!\impliedby\!)\;$
    By Lemma~\ref{lem:algo:pp:p}, we have $p' = \period{w'}$.
    Then $\varppreach{\varppidx} = \ifversion{lncs}{\linebreak}\reach{P}(p') = \ifversion{arxiv}{\makebox[0.02\textwidth]{}\linebreak}\reach{P}(\period{w'}) \ge |w'| = r-p$.
    \qed
\end{proof}

Now, we show that the invariant $\spadesuit$ always holds.

\begin{lemma}
    \label{lem:algo:pp:match}
    Throughout Algorithm~\ref{algo:pp}, we have $\substr{P}{0}{r-p}\match\substr{P}{p}{r}$.
\end{lemma}
\begin{proof}
    One must see the condition preserved at the lines in which $p$ or $r$ is updated.
    The update at Lines~\ref{algo:pp:limit_shift}--\ref{algo:pp:limit_shift_end} is trivial.
    Line~\ref{algo:pp:reach} preserves the condition, ensured by the condition of Line~\ref{algo:pp:while}.
    For Line~\ref{algo:pp:shift}, let $q:=\varppval{\varppidx}$.
    Since $q=\period{\substr{P}{0}{r-p}}$ by Lemma~\ref{lem:algo:pp:p}, we have $\substr{P}{0}{r-(p+q)} \match \substr{P}{q}{r-p} \match \substr{P}{p+q}{r}$.
    Note that Lemma~\ref{lem:algo:pp:p} requires $\spadesuit$ only at Line~\ref{algo:pp:branch}, so the argument does not circulate.
    \qed    
\end{proof}

The following lemma justifies the shift of $p$ at Lines~\ref{algo:pp:shift} and~\ref{algo:pp:limit_shift}.

\begin{lemma}
    \label{lem:ppskip}
    Consider $P\in\String$, $p\in\Nplus$ and let $r:=\reach{P}(p)$.
    Then, no prefix period $q$ of $P$ exists such that $p < q < p+\period{\substr{P}{0}{r-p}}$.
\end{lemma}
\begin{proof}
    We use Lemma~\ref{lem:shift} as $x:=\substr{P}{p}{r}$, $y:=\substr{P}{0}{r-p}$, $\delta:=q-p$ to obtain $\substr{P}{q}{r} \nmatch \substr{P}{0}{r-q}$, which means $\nhasperiod{\substr{P}{0}{r}}{q}{}$.
    Thus we have $\reach{P}(q) < r = \reach{P}(p)$, which implies that $q$ is not a prefix period of $P$ by Lemma~\ref{lem:ppreach}.
    \qed
\end{proof}

We have thus far proved the validity of Algorithm~\ref{algo:pp}.
Now, we show that the algorithm runs in $O(m)$ time.
Firstly, notice that the while loops at Line~\ref{algo:pp:outerwhile} and \ref{algo:pp:while} are repeated only $O(m)$ times in total, since the quantity $kp+r$ keeps increasing and $kp+r \le k\cdot\frac{m}{k}+m = O(m)$.
Hence, the fact we must show is that decrementing $\mathit{count}$ and $\varppidx$ at Line~\ref{algo:pp:count:dec}, \ref{algo:pp:count:reset}, and \ref{algo:pp:pointer:dec} takes $O(m)$ time in total.
As their values are always greater than or equal to their initial values, the number of decrements does not exceed the number of increments, which is $O(m)$ since they are in Line~\ref{algo:pp:while:content}--\ref{algo:pp:while:content_end}.

\begin{theorem}
    All prefix periods of $P$ and their reaches can be calculated in $O(m)$ time and $O(\log{m}+|\Params|)$ extra space.
\end{theorem}

\subsection{Searching Parameterized Matches}

\begin{algorithm}[!t]
    \DontPrintSemicolon

    \caption{\texttt{SEARCH}}
    \label{algo:search}

    \KwIn{$T,P\in\String$}
    \KwOut{all $0\le i\le |T|-|P|$ such that $\substr{T}{i}{i+|P|}\match P$}
    
    \Begin{
        $k \gets |\pcs{P}|+2$ \;
        $\mathit{first} \gets \first{P}$ \;
        $\varpp \gets \mathtt{PREFIX\_PERIODS}(P,\mathit{first})$ \;
        $\varppidx \gets -1$ \;
        $(i,j) \gets (0,0)$ \;
        Set $\mathit{count}[a] \gets 0$ for each $a\in\Params$ \;
        \While{$i < |T|-|P|$}{
            \While{$\mathtt{MATCH}(\substr{P}{0}{j-i}, \substr{T}{i}{j}, P[j-i], T[j], \mathit{first}, \mathit{count})$}{
                Increment $\mathit{count}[T[j]]$ \;
                $j \gets j + 1$ \;
                \lIf{$\varppidx+1 < |\varpp|$ and $\varppval{\varppidx+1} \le \frac{j-i}{k}$}{Increment $\varppidx$}
            }

            \vspace{1em}

            \If{$j-i = |P|$\label{algo:search:match}}{\KwSty{output} $i$}

            \vspace{1em}

            \eIf{$0 \le \varppidx < |\varpp|$ and $\varppreach{\varppidx} \ge j-i > 0$}{
                \lFor{$i \le u < i+\varppval{\varppidx}$}{Decrement $\mathit{count}[T[u]]$}
                $i \gets i + \varppval{\varppidx}$ \;
            }{
                \lFor{$i \le u < j$}{Decrement $\mathit{count}[T[u]]$}
                $i \gets i + \floor{\frac{j-i}{k}} + 1$ \;
                $j \gets i$ \;
            }
            \KwSty{until} $\varppval{\varppidx} \le \frac{j-i}{k}$ or $\varppidx=-1$ \KwSty{do} Decrement $\varppidx$ \;
        }
    }
\end{algorithm}

Our matching algorithm is shown in Algorithm~\ref{algo:search}.
As it is the case for the Galil-Seiferas algorithm, it resembles the preprocess.
Now, the invariants in Algorithm~\ref{algo:search} are obtained by replacing $p$, $r$, and $\substr{P}{p}{r}$ in Lemma~\ref{lem:algo:pp:index}--\ref{lem:algo:pp:match} with $i$, $j$, and $\substr{T}{i}{j}$, respectively.
Particularly, by the invariant that $\substr{P}{0}{j-i} \match \substr{T}{i}{j}$, one can find matching positions $i$ when $j=i+|P|$ (Line~\ref{algo:search:match}).
The shift amounts are also justified by using Lemma~\ref{lem:shift} as $x:=\substr{T}{i}{j}$ and $y:=\substr{P}{0}{j-i}$, whose conclusion $\substr{T}{i+\delta}{j} \nmatch \substr{P}{0}{j-i-\delta}$ implies $\substr{T}{i+\delta}{i+\delta+|P|} \nmatch P$ for any $\delta$ smaller than the shift by the algorithm.

\begin{theorem}
    \label{th:algo}
    The parameterized matching problem can be solved in $O(|\pcs{P}|n+m)$ time and $O(\log{m}+|\Params|)$ extra space.
\end{theorem}

\section{Conclusion and Future Work}

We studied the periodicity of parameterized strings and extended the Galil-Seiferas algorithm~\cite{GSlog} for parameterized matching.
The proposed algorithm requires only sublinear extra space.
The properties of periods of parameterized strings we presented in this paper may be used to design more space-efficient algorithms for parameterized matching, as Galil and Seiferas~\cite{GS} used prefix periods to design a constant-extra-space algorithm for exact matching.

\ifversion{lncs}{\bibliographystyle{splncs04}}
\ifversion{arxiv}{\bibliographystyle{plain}}
\bibliography{paper}

\end{document}